\documentclass[conference]{IEEEtran}
\IEEEoverridecommandlockouts

\usepackage{amsmath,amssymb,amsfonts,amsthm}
\usepackage{algorithm}
\usepackage{algpseudocode}
\usepackage{booktabs}
\usepackage{csquotes}
\usepackage{cite}
\usepackage{enumerate}
\usepackage[figuresleft]{rotating}
\usepackage[autolanguage]{numprint}
\usepackage{graphicx}
\usepackage[utf8]{inputenc}
\usepackage[locale=US]{siunitx}
\usepackage{subcaption}
\usepackage{textcomp}
\usepackage[disable]{todonotes}
\usepackage{url}
\usepackage{tikz}
\usepackage{xcolor}
\usepackage{xspace}

\newcommand{\ie}{i.\,e.,\xspace}
\newcommand{\eg}{e.\,g.,\xspace}

\newcommand{\etal}{et al.\xspace}

\newcommand{\argmin}{\operatorname{argmin}}
\newcommand{\argmax}{\operatorname{argmax}}

\newcommand{\exchange}[3]{#1^{+#2}_{-#3}}
\newcommand{\add}[2]{#1^{+#2}}
\newcommand{\remove}[2]{#1_{-#2}}
\newcommand{\dist}[1]{\mathrm{dist}(#1)}

\newcommand{\onedigit}[1]{\num[round-mode=places,round-precision=1]{#1}}

\newcommand{\gsMaxSpeed}{927.871445138909}
\newcommand{\gsMinSpeed}{42.96954466604031}

\newcommand{\gsPSevenFiveKQualNum}{99.38257368840114}
\newcommand{\gsPSevenFiveKSpeed}{127.83319715440305}

\newcommand{\minSpeedupKHundred}{42.96954466604031}
\newcommand{\gsLocalSpeedupOverNonLocal}{3.1341149912690742}
\newcommand{\gsLocalSpeedKTen}{365.83696556394904}
\newcommand{\gsLocalQualKTen}{92.0747944022351}

\newcommand{\wGSQualInc}{12.384054953251923}
\newcommand{\wGSSpeed}{793.6424400458743}

\newtheorem{definition}{Definition}
\newtheorem{lemma}{Lemma}

\begin{document}

\title{Local Search for Group Closeness Maximization \\ on Big Graphs
\thanks{This work is partially supported by German Research Foundation (DFG) grant ME 3619/3-2
within Priority Programme 1736 \textit{Algorithms for Big Data}.}
}

\author{
\IEEEauthorblockN{
Eugenio Angriman,
Alexander van der Grinten,
Henning Meyerhenke}
\IEEEauthorblockA{
Department of Computer Science\\
Humboldt-Universit\"at zu Berlin, Germany\\
{\{angrimae, avdgrinten, meyerhenke\}@hu-berlin.de}}
}

\maketitle
\IEEEpubid{\begin{minipage}{\textwidth}\ \\[12pt] \centering
\copyright\ 2019 IEEE.
Personal use of this material is permitted. Permission from IEEE must be
obtained for all other uses, in any current or future media, including
reprinting/republishing this material for advertising or promotional purposes, creating new
collective works, for resale or redistribution to servers or lists, or reuse of any copyrighted
component of this work in other works.\hfill
\end{minipage}}

\IEEEpubidadjcol

\begin{abstract}
In network analysis and graph mining, closeness centrality is a popular measure to infer the
importance of a ver\-tex. Computing closeness efficiently for individual vertices
received considerable attention. The $\mathcal{NP}$-hard problem of \emph{group closeness maximization},
in turn, is more challenging: the objective is to find a vertex \emph{group} that is central
\emph{as a whole} and state-of-the-art heuristics for it do not scale to
very big graphs yet.

In this paper, we present new local search heuristics
for group closeness maximization.
By using randomized approximation techniques and dynamic data structures,
our algorithms are often able to perform locally optimal decisions efficiently.
The final result is a group with high (but not optimal) closeness centrality.

We compare our algorithms to the current state-of-the-art
greedy heuristic both on weighted and on unweighted real-world graphs.
For graphs with hundreds of millions of edges,
our local search algorithms take only around ten minutes, while greedy requires more than ten hours.
Overall, our new algorithms are between one and two orders of magnitude faster,
depending on the desired group size and solution quality.
For example, on weighted graphs and $k = 10$, our algorithms yield solutions
of $\onedigit{\wGSQualInc}\%$ higher quality, while also being
$\onedigit{\wGSSpeed}\times$ faster.
For unweighted graphs and $k = 10$, we achieve solutions within
$\onedigit{\gsPSevenFiveKQualNum}\%$ of the state-of-the-art quality
while being $\onedigit{\gsPSevenFiveKSpeed}\times$ faster.
\end{abstract}


\section{Introduction}
\label{sec:intro}
Identifying important vertices in large networks is one of the main problems in
network analysis~\cite{newman2018networks}.
For this purpose, several centrality measures have been introduced over the
past decades~\cite{boldi2014axioms}.
Among them, one of the most widely-used measures is closeness~\cite{bavelas1948mathematical}.
For a given vertex $v$, it is defined as the reciprocal of the average shortest-path distance
from $v$ to all other vertices.
The problem of identifying the $k$ vertices with highest closeness centrality has
received significant attention~\cite{bergamini2019computing,okamoto2008ranking,olsen2014efficient}.
In graph mining applications, however, it is often necessary to
determine a \emph{group} of vertices that is central \emph{as a whole} --
which is an orthogonal problem shown to be $\mathcal{NP}$-hard~\cite{chen2016efficient}.
One can view group closeness as special case (on graphs) of the well-known $k$-median
problem for facility location. Example applications include:
(i) retailers that want to advertise their product via social media; they could
select as promoters the group of $k$ members with highest centrality ($\approx$ influence over
the other members)~\cite{zhu2014maximizing}; (ii) in P2P networks, shared resources could be placed on
$k$ peers so that they are easily accessible by others~\cite{gkantsidis2004random};
(iii) in citation networks, group centrality measures can be employed
as alternative indicators for the influence of journals or papers
within their field~\cite{leydesdorff2007betweenness}.

Everett and Borgatti~\cite{everett1999centrality} formalized the concept of
centrality for groups of vertices; the closeness of a group $S$ is defined
as the reciprocal of the average distance from $S$ to all other vertices of the graph.
While exact algorithms to find a group with maximal group closeness are known
(\eg algorithms based on integer linear programming (ILP)~\cite{bergamini2018scaling}),
they do not scale to graphs with more than a few thousand edges.
Hence, in practice, heuristics
are used to find groups with high group closeness on large real-world data sets.
For example, Chen~\etal~\cite{chen2016efficient} proposed a greedy algorithm
that heuristically
computes a group $S$ with high closeness $C(S)$. To obtain a group of size $k$, it performs
$k$ iterations and adds in each iteration the vertex with highest marginal contribution to $S$.\footnote{While this greedy algorithm was claimed to have
	a bounded approximation quality
	(\eg in~\cite{chen2016efficient,bergamini2018scaling}),
	the proof of this bound relied on the assumption that $C$ is submodular.
	A recent update~\cite{DBLP:journals/corr/abs-1710-01144} to the conference
	version~\cite{bergamini2018scaling} revealed that, in fact, $C$ is \emph{not} submodular.
	We are not aware of any approximation algorithm for group closeness that scales to large graphs.}
It was shown empirically that this greedy algorithm yields solutions of very high quality (within $97\%$ of the optimum)~\cite{bergamini2018scaling}
-- at least for those small graphs where running a comparison against an exact algorithm is still
feasible within a few hours.
Due to greedy's promising quality, Bergamini~\etal~\cite{bergamini2018scaling} proposed techniques
to speed up the algorithm by pruning, \eg by exploiting the supermodularity of farness, \ie $1/C(S)$.
Pruning yields a significant acceleration (while retaining the same solution);
however, graphs with hundreds of millions of edges still require
several hours to complete.
Indeed, pruning is most effective when the group is already large.
When performing the first addition, however,
the greedy algorithm has to perform one (pruned) single-source shorted-path
(SSSP) computation \emph{for each vertex in the graph} to compute its marginal contribution,
and this phase scales super-linearly in the size
of the graph.

\IEEEpubidadjcol
\paragraph*{Our Contribution}
We present two new local search algorithms for group closeness: the first algorithm,
\emph{Local-Swap}, requires little time per iteration but can only
exchange vertices locally.
Our second algorithm, \emph{Grow-Shrink}, is able to perform non-local
vertex exchanges, but updating its data structures is more involved.
The final result of both algorithms is heuristic, \ie no approximation guarantee is known.
Yet, each iteration of Local-Swap maximizes (in approximation) a lower bound on the objective function,
while each iteration of Grow-Shrink locally maximizes the objective function itself.
Despite these favorable properties,
the time complexity of a single iteration of our algorithms
matches the time complexity of a single evaluation of the objective function,
\ie for unweighted graphs, it is linear in the size of the graph.

Our experiments show that our best algorithm,
\emph{extended Grow-Shrink}, finds solutions with a closeness score
greater than $\onedigit{\gsPSevenFiveKQualNum}\%$ of the score of a greedy solution,
while being $\onedigit{\gsPSevenFiveKSpeed}\times$ faster to compute ($k = 10$).
We see this algorithm as our main contribution for unweighted graphs.
When quality is not a primary concern, our other algorithms
can further accelerate the computation:
For example, the non-extended variant of
Grow-Shrink yields solutions for groups of size 10
whose quality is $\onedigit{91.1}\%$ compared to the state of the art;
in this case, it is $\onedigit{700.2}\times$ faster.
The speedup varies between $\onedigit{\gsMaxSpeed}\times$
and $\onedigit{\gsMinSpeed}\times$ for groups of sizes 5 and 100,
respectively.

On weighted graphs, our algorithms even improve both the quality and
the running time performance compared to the state of the art,
returning solutions of $\onedigit{\wGSQualInc}\%$ higher quality at
a speedup of $\onedigit{\wGSSpeed}\times$ ($k = 10$).
Other trade-offs between quality and performance are possible
and discussed in Section~\ref{sec:experiments}.

\section{Local Search for Group Closeness}
\label{sec:local_search}

Let $G = (V, E)$ be an undirected connected graph.
We allow both unweighted and positively weighted graphs $G$.
Subsets $S \subseteq V$ are called groups.
The \emph{farness} of any given group $S$ is defined as:
\[ f(S) = \sum_{v \in V \setminus S} \mathrm{dist}(S, v). \]
Here, $\dist{S, v}$ refers to the minimal shortest-path distance from any $s \in S$
to $v$ in $G$.
Furthermore, the (group) closeness
of $S$ is defined as
$C(S) = |V|/f(S)$,
\ie $C(S)$ is the reciprocal of the average distance of all vertices in $V \setminus S$
to the nearest vertex in $S$.
Recall that determining the group $S^*$ that maximizes $C$
over all groups $S$ with $|S| \leq k$
is known to be $\mathcal{NP}$-hard~\cite{chen2016efficient};
we are not aware of any algorithm with
a bounded approximation ratio.

We consider the problem of improving
the group closeness
of a given set $S$ via local search.
More precisely, we consider \emph{exchanges} of vertices from $S$ and $V \setminus S$.
Let $u$ be a vertex in $S$ and $v \in V\setminus S$.
To simplify the presentation, we use the notation
$\exchange{S}vu := (S\setminus\{u\}) \cup \{v\}$ to denote the set
that is constructed by exchanging $u$ and $v$.
We also use the notation $\add{S}v := S \cup \{v\}$ and $\remove{S}u := S \setminus \{u\}$,
to denote
vertex additions and removals, respectively.

Note that, as our algorithms can only perform vertex exchanges,
they require the construction of an initial set $S$ before the algorithms start.
To avoid compromising our algorithms' running times, we cannot afford
a superlinear initialization step.
Thus, in all of our local search algorithms,
we simply choose the initial $S$ uniformly at random.
For large graphs, this initialization can be expected to cover the
graph reasonably well. Exploratory experiments revealed that
other obvious initialization techniques (such as selecting the
$k$ vertices with highest degree) did not improve the
performance of the algorithm.

\subsection{Estimating the Quality of Vertex Exchanges}
\label{sec:dagsize}
It is known that a simple greedy ascent algorithm yields results of
good quality on real-world graphs~\cite{bergamini2018scaling}.
This greedy algorithm starts with an empty set $S$ and iteratively
adds vertices $v \in V \setminus S$ to $S$ that maximize
$f(S) - f(\add{S}v)$. Depending on the input graph and the value
of $k$, however, the greedy algorithm might need to evaluate
the difference in $f$ for a substantial number of vertices
-- this computation is rather expensive for large real-world graphs.

The algorithms in this paper aim to improve upon the running
time of the greedy algorithm.
We achieve this by considering only \emph{local} vertices for $v$,
\ie vertices that are already \enquote{near} $S$.
It is clear that selecting only local vertices would
decrease the quality of a greedy solution (as the greedy
algorithm does not have
the ability to eventually correct suboptimal choices).
However, this is not necessarily
true for our algorithms based on vertex exchanges in Sections~\ref{sec:swapping}
and~\ref{sec:alternating}.

To make our notion of locality more concrete,
let $B_S \subseteq G$ be the DAG constructed by running a SSSP
algorithm (\ie BFS or Dijkstra's algorithm) from the vertices
in $S$. We remark that we work with the full SSSP DAG here (and not a SSSP tree).
Here, the vertices of $S$ are all considered as sources of the SSSP algorithm,
\ie they are at distance zero.
Furthermore, define
\begin{align}
	\nonumber
	\Delta^-(v)
	&:= f(S) - f(\add{S}v)\\
	\label{eq:delta-minus}
	&\phantom{:}= \sum_{x \in V \setminus S} \mathrm{dist}(S, x) - \mathrm{dist}(\add{S}v, x).
\end{align}
To compute the greedy solution, it seems to be necessary to compute
$\Delta^-(v)$ exactly for a substantial number of vertices
$v$.\footnote{The techniques of~\cite{bergamini2018scaling}
can avoid some of the computations; nevertheless, many evaluations
of $\Delta^-(v)$ still have to be performed.}
As discussed above, this seems to be impractical for large graphs.
However, a lower bound for $\Delta^-(v)$ can
be computed from the shortest path DAG $B_S$.
To this end, let $D_v$ be the set of vertices reachable from $v$ in $B_S$.
\begin{lemma}
	\label{lemma:additions}
	It holds that:
	\[ \Delta^-(v) \geq |D_v|\cdot \mathrm{dist}(S, v). \]
	In the unweighted case, equality holds if $v$ is a neighbor of $S$.
\end{lemma}
This lemma can be proven from the definition
of $B_S$. A formal proof can be found in Appendix~\ref{apx:proofs}.

The bound of Lemma~\ref{lemma:additions} will be used in the two algorithms in
Sections~\ref{sec:swapping} and ~\ref{sec:alternating}.
Instead of picking vertices that maximize $\Delta^-$, those
algorithms pick vertices that maximize the right-hand side of Lemma~\ref{lemma:additions},
\ie $|D_v|\cdot \dist{S, v}$.
The bound is local in the sense that it is more accurate for vertices
near $S$: in particular, the reachability sets of
vertices in $(V \setminus S) \cap N(S)$ are larger in $G$ than
those in $B$, as $B$ does not contain back edges.
Unfortunately, computing the size of $D_v$ exactly for all $v$
still seems to be
prohibitively expensive: indeed, the fastest known algorithm to compute
the size of the transitive closure of a DAG (= $\sum |D_v|$)
relies on iterated (Boolean) matrix multiplication
(hence, the best known exact algorithm has a complexity
of $\mathcal{O}(n^{2.37})$~\cite{le2014powers}).
However, it turns out that randomized algorithms
can be used to approximate the sizes of
$D_v$ for all $v$ \emph{at the same time}.
We employ the randomized algorithm of Cohen~\cite{cohen1997size} for this task.
In multiple iterations, this algorithm samples a random number for each vertex
of the graph $G$, accumulates in each vertex $v$ the minimal
random number of any vertex reachable from $v$, and
estimates $|D_v|$ based on this information.

We remark that since Cohen's algorithm yields an approximation, but not a lower bound for the
right-hand side of Lemma~\ref{lemma:additions}, the inequality of the Lemma
can be violated in our algorithms; in particular, it can happen that our algorithms pick a vertex
$v$ such that $\Delta^-(v) < 0$. In this case, instead of decreasing the closeness centrality
of our current group, our algorithms terminate. Nevertheless, our experiments
demonstrate that on real-world instances, a coarse approximation of the reachability
set size is enough for Lemma~\ref{lemma:additions} to yield useful candidates
for vertex exchanges (see Section~\ref{sub:reachability_engineering}
and Appendix~\ref{apx:samples_exp}).

\subsection{Local-Swaps Algorithm}
\label{sec:swapping}

\begin{algorithm}[t]
\caption{Overview of Local-Swaps Algorithm}
\label{algo:swaps}
\begin{algorithmic}[1]
  \Repeat
  	\State approximate $|D_v|$ for all $V \setminus S$
	\State \label{line:pick-swap} $(u, v) \gets
		\argmax_{u \in S, v \in N(u) \setminus S} |D_v| - |\Lambda_u|$
    \State $S \gets \exchange{S}vu$
	\State \label{line:recompute-swap} run pruned BFS from $v$ \Comment{to recompute $f(S)$}
  \Until{previous iteration did not decrease $f(S)$}
\end{algorithmic}
\end{algorithm}

Let us first focus on unweighted graphs.
To construct a fast local search algorithm,
a straightforward idea is to allow \emph{swaps} between
vertices in $S$ and their neighbors in $V\setminus S$.
This procedure can be repeated until no swap can decrease $f(S)$.
Let $u \in S$ be a vertex of the group and let $v \in N(S) \setminus S$ be
one of its neighbors outside of the group.
To determine whether swapping $u$ and $v$ (\ie replacing $S$ by $\exchange{S}vu$)
is beneficial, we have to check whether $f(S) - f(\exchange{S}vu) > 0$,
\ie whether the farness is decreased by the swap.
The challenge here is to find a pair $u, v$ that satisfies this inequality
(without checking all pairs $u, v$ exhaustively)
and to compute the difference $f(S) - f(\exchange{S}vu)$ quickly.
Note that a crucial ingredient that allows us to construct an
efficient algorithm is that the distance of $S$ to every vertex $x$
can only change by $\pm 1$ when doing a swap.
Hence, we only have to count the numbers of vertices where
the distance changes by $-1$ and the number of vertices
where it changes by $+1$.

Our algorithm requires a few auxiliary data structures to compute $f(S) - f(\exchange{S}vu)$.
In particular, we store the following:
\begin{itemize}
	\item the distance $\dist{S, x}$ from $S$ to all vertices $x \in V \setminus S$,
	\item a set $\lambda_x := \{w \in S : \dist{S, x} = \dist{w, x}\}$
		for each $x \in V \setminus S$
		that contains all vertices in $S$ that realize
		the shortest distance from $S$ to $x$,
	\item the value $|\Lambda_w|$ for each $w \in S$,
		where $\Lambda_w := \{x \in V \setminus S : \lambda_x = \{w\}\}$
		is the set of vertices for which the shortest distance is
		realized exclusively by $w$.
\end{itemize}

Note that the sets $\lambda_x$ consume $\mathcal{O}(k |V|)$ memory in total.
However, since $k \ll |V|$, this can be afforded even for large real-world graphs.
In our implementation, we store each $\lambda_x$ in only $k$ bits.

All of those auxiliary data structures can be
maintained dynamically during the entire algorithm
with little additional overhead.
More precisely, after a $u$-$v$ swap is done,
$v$ is added to all $\lambda_x$ satisfying $\dist{v, x} = \dist{S, x}$;
for $x \in V \setminus S$ that satisfy $\dist{v, x} < \dist{S, x}$,
the set $\{v\}$ replaces $\lambda_x$.
$u$ can be removed from all $\lambda_x$ by a linear scan through
all $x \in V \setminus S$.

Algorithm~\ref{algo:swaps} states a high-level overview of the algorithm.
In the following, we discuss how to pick a good swap (line~\ref{line:pick-swap}
of the pseudocode)
and how to update the data structures after a swap (line~\ref{line:recompute-swap}).
The running time of the algorithm is dominated by the initialization
of $\lambda_x$. Thus, it runs in $\mathcal{O}(k |V| + |E|)$ time for each update.

\subsubsection{Choosing a Good Swap}
\label{sec:choose-swap}
Because it would be too expensive to compute the exact difference in $f$ for each
possible swap, we find the pair of vertices $(u, v)$
with $u \in S$, $v \in N(v) \setminus S$ that
maximizes $|D_v|\cdot \dist{S, v} - |\Lambda_u| = |D_v| - |\Lambda_u|$.
Note that this value is a lower bound for the
decrease of $f$ after swapping $u$ and $v$:
In particular, Lemma~\ref{lemma:additions} implies that
$|D_v|$ is a lower bound for the decrease in farness
when adding $v$ to $S$. Additionally, $|\Lambda_u|$ is an upper bound
for the increase in farness when removing $u$ from $S$
(and hence also for the increase in farness when removing $u$ from $\add{S}v$).\footnote{
Note, however, that this bound is trivial
if $|D_v| - |\Lambda_u| \leq 0$.}
Thus, we can expect this strategy to yield pairs of vertices
that lead to a decrease of $f$.
To maximize $|D_v| - |\Lambda_u|$,
for each $v \in V \setminus S$, we
compute the neighbor $u \in N(v) \cap S$ that minimizes
$|\Lambda_u|$ (in $\mathcal{O}(|V| + |E|)$ time). Afterwards, we can maximize
$|D_v| - |\Lambda_u|$ by a linear scan over all $v \in V \setminus S$.

\subsubsection{Computing the Difference in Farness}
Instead of comparing distances, it is sufficient to
define sets of vertices whose distance to $S$ is increased or decreased
(by 1) by the swap:
\begin{align*}
	H^+_{u,v} &:= \{x \in V : \dist{S, x} < \dist{\exchange{S}vu, x} \},\\
	H^-_{u,v} &:= \{x \in V : \dist{S, x} > \dist{\exchange{S}vu, x} \}
\end{align*}
As $\dist{S, x} - \dist{\exchange{S}vu, x} \in \{-1,0,+1\}$, it holds that:
\begin{lemma}
	$f(S) - f(\exchange{S}vu) = |H^-_{u,v}| - |H^+_{u,v}|$
\end{lemma}

Fortunately, computing $H^-_{u,v}$ is straightforward: this can be done by running
a BFS rooted at $v$; the BFS simply counts those vertices $x$
for which $\dist{v, x} < \dist{S, x}$.
Hence, to check this condition, we have to
store the values of $\dist{S, x}$ for all $x \in V$.
We remark that it is not necessary to run a full BFS:
indeed, we can prune the search once $\dist{v, x} \geq \dist{S, x}$
(\ie if the BFS is about to visit a vertex $x$ satisfying this condition,
the search continues without visiting $x$).
However, as we will see in the following, it makes sense to sightly relax
the pruning condition and only prune
the BFS if $\dist{v, x} > \dist{S, x}$; this allows us to update
our auxiliary data structures on the fly.

$H^+_{u,v}$ can be computed from $|\Lambda_u|$ with the help
of the auxiliary data structures.
We note that $H^+_{u,v} \subseteq \Lambda_u$,
as only vertices $x$ where $\dist{S, x}$ is uniquely realized
by $u$ (out of all vertices in the group) can have their distance from $S$ increased
by the swap.
As $H^+_{u,v} \cap H^-_{u,v} = \emptyset$,
we can further restrict this inclusion to
$H^+_{u,v} \subseteq \Lambda_u \setminus H^-_{u,v}$,
but, in general, $\Lambda_u \setminus H^-_{u,v}$ will consist of more vertices
than just $H^+_{u,v}$. More precisely, $\Lambda_u \setminus H^-_{u,v}$ can be partitioned
into $\Lambda_u \setminus H^-_{u,v} = H^+_{u,v} \cup L^0_{u,v}$,
where $L^0_{u,v} := \{x \in \Lambda_u : \dist{u, x} = \dist{v, x}\}$ consists
only of vertices whose distance is neither increased nor decreased by the swap.
By construction, $L^0_{u,v}$ and $H^+_{u,v}$ are disjoint.
This proves that the following holds:
\begin{lemma}
	$|H^+_{u, v}| = |\Lambda_u| - |\Lambda_u \cap H^-_{u,v}| - |L^0_{u,v}|$.
\end{lemma}
We note that $L^0_{u,v}$ (and also $\Lambda_u \cap H^-_{u,v}$)
is completely visited by our BFS. To determine
$|L^0_{u,v}|$, the BFS only has to count the vertices $x$ that satisfy
$\dist{v,x} = \dist{S,x}$ and $\lambda_x = \{u\}$.
On the other hand, to determine $|\Lambda_u \cap H^-_{u,v}|$, it has to count
the vertices $x$ satisfying $\dist{v,x} < \dist{S,x}$ and $\lambda_x = \{u\}$.

\subsection{Grow-Shrink Algorithm}
\label{sec:alternating}

\begin{algorithm}[t]
\caption{Overview of Grow-Shrink Algorithm}
\label{algo:grow-shrink}
\begin{algorithmic}[1]
	\Repeat
		\State \label{line:grow-start} approximate $|D_v|$ for all $v \in V \setminus S$
		\State $v \gets \argmax_{v \in V \setminus S} |D_v|\cdot d(v)$
		\State $S \gets \add{S}v$
		\State \label{line:grow-end} run pruned BFS from $v$ \Comment{to recompute $f(S)$, $d$, $d'$}
		\State \label{line:shrink-start} $u \gets \argmin_{u \in S} \sum_{x \in R_u} d'(x) - d(x)$
		\State $S \gets \remove{S}u$
		\State \label{line:shrink-end} \label{line:dijkstra} run Dijkstra-like algo. \Comment{to recompute $f(S)$, $d$, $d'$}
	\Until{previous iteration did not decrease $f(S)$}
\end{algorithmic}
\end{algorithm}

The main issue with the swapping algorithm from Section~\ref{sec:swapping}
is that it can only exchange a vertex $u \in S$
with one of its neighbors $v \in N(u) \setminus S$.
Due to this behavior, the algorithm might take a long time to converge
to a local optimum. It also makes it hard to escape a local optimum:
indeed, the algorithm will terminate if no swap with a neighbor
can improve the closeness.

Our second algorithm lifts those limitations.
It also allows $G$ to be a weighted graph.
In particular, it allows vertex exchanges that change the
distances from $S$ to the vertices in $V \setminus S$ by
arbitrary amounts.
Computing the exact differences $f(S) - S(\exchange{S}vu)$
for all possible pairs of $u$ and $v$ seems to be impractical
in this setting.
Hence, we decompose the vertex exchange of $u$ and $v$
into two operations: the addition of $v$ to $S$ and the removal of $u$ from $\add{S}v$.
In particular, we allow the set $S$ to grow to a size of $k + 1$
before we shrink the size of $S$ back to $k$. Thus, the cardinality
constraint $|S| \leq k$ is temporarily violated; eventually,
the constraint is restored again.
Fortunately, the individual differences $f(S) - f(\add{S}v)$
and $f(\add{S}v) - f(\exchange{S}vu)$
(or bounds for those differences)
turn out to be efficiently computable
for all possible $u$ and $v$,
at least in approximation. We remark, however, that while this
technique does find the vertex that maximizes $f(S) - f(\add{S}v)$
and the vertex that maximizes $f(\add{S}v) - f(\exchange{S}vu)$,
it does \emph{not} necessarily find the \emph{pair} of vertices maximizing
$f(S) - f(\exchange{S}vu)$.
Nevertheless, our experiments in Section~\ref{sec:experiments} demonstrate
that the solution quality of this algorithm is superior to the quality of the
local-swaps algorithm.

In order to perform these computations, our algorithm maintains
the following data structures:
\begin{itemize}
	\item the distance $d(x)$ of each vertex $x \notin S$ to $S$,
		and a representative $r(x) \in S$ that realizes this
		distance, \ie it holds that $\mathrm{dist}(S, x) = \mathrm{dist}(r(x), x) = d(x)$,
	\item the distance $d'(x)$ from $S \setminus \{r(x)\}$ to $x$
		and representative $r'(x)$ for this distance (satisfying
		the analogous equality).
\end{itemize}

Since the graph is connected, these data structures are well-defined
for all groups $S$ of size $|S| \geq 1$. Furthermore, the difference
between $d'(x)$ and $d(x)$ yields exactly the difference in farness
when $r(x)$ is removed from the $S$. Later, we will use this fact
to quickly determine differences in farness.

We remark that it can happen that $d(x) = d'(x)$; nevertheless,
$r(x)$ and $r'(x)$ are always
distinct. Indeed, there can be two
distinct vertices $r(x)$ and $r'(x)$ in $S$ that
satisfy $\mathrm{dist}(r(x), x) = \mathrm{dist}(r'(x), x) = \mathrm{dist}(S, x)$.
With $R_u$ and $R'_u$, we denote the set of vertices
$x \in V \setminus S$ with $r(x) = u$ and $r'(x) = u$, respectively.

Algorithm~\ref{algo:grow-shrink} gives a high-level overview of the algorithm.
In the following two subsections, we discuss the growing phase
(line \ref{line:grow-start}-\ref{line:grow-end})
and the shrinking phase (line \ref{line:shrink-start}-\ref{line:shrink-end}) individually.
The running time of Grow-Shrink is dominated by the Dijkstra-like algorithm
(in line~\ref{line:dijkstra}). Therefore, it runs in $\mathcal{O}(|V| + |E| \log |V|)$ time
per update
(when using an appropriate priority queue).
The space complexity is $\mathcal{O}(|V| + |E|)$.

\subsubsection{Vertex additions}
\label{sec:growing}
When adding a vertex $v$ to $S$,
we want to select $v$ such that $f(\add{S}v)$ is minimized.
Note that minimizing $f(\add{S}v)$ is equivalent to maximizing the difference
$f(S) - f(\add{S}v) = \Delta^-(v)$.
Instead of maximizing $\Delta^-(v)$, we maximize the lower
bound $|D_v| \cdot \dist{S, v}$.
We perform a small number of iterations of the
reachability set size approximation algorithm (see Section~\ref{sec:dagsize}) to select
the vertex $v$ with (approximatively) largest $|D_v|$.

After $v$ is selected, we perform a BFS from $v$ to compute $\Delta^-(v)$ exactly.
As we only need to visit the vertices whose distance to $\add{S}v$ is
smaller than to $S$, this BFS can be pruned
once a vertex $x$ is reached with $\mathrm{dist}(S, x) < \mathrm{dist}(v, x)$.
During the BFS, the values of $d, d', r, r'$ are updated to reflect the
vertex addition: the only thing that can happen here is that $v$
realizes either of the new distances $d$ or $d'$.

\subsubsection{Vertex removals}
\label{sec:shrinking}
For vertex removals, we can efficiently calculate the exact increase
$\Delta^+(u) := f(\remove{S}u) - f(S)$ in farness for all vertices $u \in S$,
even without relying on approximation.
In fact, $\Delta^+(u)$ is given as:
\[ \Delta^+(u) = \sum_{x \in R_u} d'(x) - d(x). \]
We need to compute $k$ such sums (\ie $\Delta^+(u)$ for each $u \in S$);
but they can all be computed at the same time by a single linear scan through all
vertices $x \in V$.

\begin{figure}
\centering

\begin{tikzpicture}
\small

\coordinate (r1) at (2,2.4);
\coordinate (r2) at (6.25,2.85);
\coordinate (r3) at (8,2.8);

\draw[cyan, rounded corners=30pt,dashed,thick]
  (0,1.1) rectangle ++(4,2.2);
\draw[orange, rounded corners=40pt,dashed,thick]
  (4,1.3) rectangle ++(3.5,3);
\node[cyan, draw, shape=rectangle, minimum width=2.55cm, minimum height=1.5cm, anchor=center,rounded corners=20pt,thick] at (r1) {};
\node [orange, draw, shape=rectangle, minimum width=2.5cm, minimum height=2.0cm, anchor=center,rounded corners=25pt,thick] at  (r2) {};
\node [cyan,draw, shape=rectangle, minimum width=1cm, minimum height=1.3cm, anchor=center,rounded corners=10pt,
thick] at (r3) {};

\draw (r1) circle [radius=3pt] coordinate (w);
\draw (w) node [below=.1cm] {$w$};
\draw (r1) ++(.75,.3) circle [radius=3pt] coordinate(w1);
\draw (w1) ++(.85,-.3) circle [radius=3pt] coordinate(w2);
\draw (w2) ++(.95,0.4) circle [radius=3pt] coordinate(w3);

\path (w) ++(30:3pt) coordinate (pw1);
\path (w1) ++(180:3pt) coordinate (pw2);
\path (w1) ++(0:3pt) coordinate (pw3);
\path (w2) ++(160:3pt) coordinate (pw4);
\path (w2) ++(0:3pt) coordinate (pw5);
\path (w3) ++(185:3pt) coordinate (pw6);

\draw[-] (pw1) -- (pw2);
\draw[-] (pw3) -- (pw4);
\draw[densely dotted, ultra thick] (pw5) -- (pw6);

\draw (r2) circle [radius=3pt];
\draw (r2) node [above=.1cm] {$u$};
\draw (r2) ++(.75,-.45) circle [radius=3pt] coordinate (u1);
\draw (r3) circle [radius=3pt];
\draw (r3) node [below=.1cm] {$w'$};

\path (r2) ++(325:3pt) coordinate (pu1);
\path (u1) ++(145:3pt) coordinate (pu2);
\path (u1) ++(30:3pt) coordinate (pu3);
\path (r3) ++(190:3pt) coordinate (pu6);

\draw[-] (pu1) -- (pu2);
\draw[dashed, ultra thick] (pu3) -- (pu6);

\coordinate (l1) at (.1, 4.2);
\coordinate (l2) at (.55, 4.2);
\coordinate (l3) at (.1, 3.9);
\coordinate (l4) at (.55, 3.9);
\draw[densely dotted, ultra thick] (l1) -- (l2);
\draw[dashed, ultra thick] (l3) -- (l4);
\node [right=.1cm of l2,align=left] {$d'$-boundary};
\node [right=.1cm of l4,align=left] {$d$-boundary};
\end{tikzpicture}
\caption{$w$, $u$ and $w'$ are vertices in $S$.
Vertices within the solid regions
belong to $R_w$, $R_u$ and $R_{w'}$, respectively.
Vertices within the dashed regions belong
to $R'_w$ and $R'_u$, respectively. After removing $u$ from $S$,
the vertices in $R'_u$ will have an invalid $r'$ and $d'$.}
\label{fig:boundary_nodes}
\end{figure}
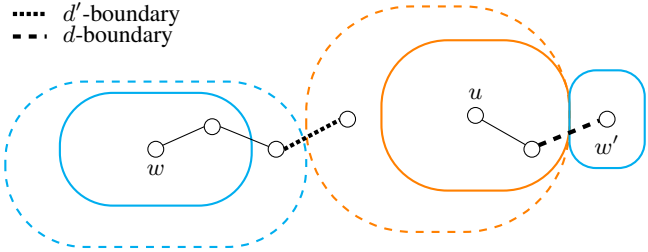

It is more challenging, however, to update $d, d', r$ and $r'$
after removing a vertex $u$ from $S$.
For vertices $x$ with an invalid $d(x)$ (\ie vertices in $R_u$),
we can simply update $d(x) \gets d'(x)$ and $r(x) \gets r'(x)$.
This update invalidates $d'(x)$ and $r'(x)$. In the
following, we treat $d'(x)$ as infinite and $r'(x)$ as
undefined for all updated vertices $x$; eventually those expressions
will be restored to valid values using the algorithm that we
describe in the remainder of this section.
Indeed, we now have to handle all vertices with an invalid $d'(x)$
(\ie those in $R_u \cup R'_u$).
This computation is more involved. We run a
Dijkstra-like algorithm (even in the unweighted case)
to fix $d'(x)$ and $r'(x)$.
The following definition yields the starting points for our
Dijkstra-like algorithm.

\begin{definition}
	Let $x \in V$ be any vertex and let $y \in N(x) \cap (R_u \cup R'_u)$
	be a neighbor of $x$ that needs to be updated.
	\begin{itemize}
		\item We call $(x, y)$ a \emph{$d$-boundary pair for $y$} iff $r(x) \neq r(y)$.
			In this case, we set $b(x, y) := d(x) + \mathrm{dist}(x, y)$.
		\item We call $(x, y)$ a \emph{$d'$-boundary pair for $y$} iff
			$r(x) = r(y)$ and $x \notin R_u \cup R'_u$.
			In this case, we set $b(x, y) := d'(x) + \mathrm{dist}(x, y)$.
	\end{itemize}
	In both cases, $b(x, y)$ is called the \emph{boundary distance}
	of $(x, y)$.
\end{definition}

The definition is illustrated in Figure~\ref{fig:boundary_nodes}.
Intuitively, boundary pairs define the boundary between regions of $G$ that
have a valid $d'(x)$ (blue regions in Figure~\ref{fig:boundary_nodes})
and regions of the graph that have an invalid $d'(y)$
(orange region in Figure~\ref{fig:boundary_nodes}).
The boundary distance $b(x, y)$
corresponds to the value of $d'$ that a SSSP algorithm could
propagate from $x$ to $y$.
We need to distinguish $d$-boundary pairs and $d'$-boundary pairs
as the boundary distance can either be propagated on a shortest path
from $S$ over $x$ to $y$ (in case of a $d$-boundary pair) or on a shortest
path from $\remove{S}{r(x)}$ over $x$ to $y$ (in case of a $d'$-boundary pair).

Consider all $y \in V \setminus S$ such that
there exists at least one \mbox{($d$- or $d'$-)}boundary pair for $y$.
For each such $y$,
let $(x, y)$ be the boundary
pair with minimal boundary distance $b(x, y)$.
Our algorithm
first determines all such $y$ and updates ${d'(y) \gets b(x, y)}$.
If $(x, y)$ is a $d$-boundary pair, we set $r(y) \gets r(x)$;
for $d'$-boundary pairs, we set $r(y) \gets r'(x)$.
After this initial update, we run a Dijkstra-like algorithm starting from
these vertices $y$ for which a boundary pair exists.
This algorithm treats $d'$ as the distance.
Compared to the usual Dijkstra algorithm, our algorithm needs the following
modifications:
For each vertex $x$, our algorithm only visits those neighbors $y$
that satisfy $r(y) \neq r'(x)$. Furthermore, whenever such a visit
results in an update of $d'(y)$,
we propagate $r'(y) \gets r'(x)$.
Note that these conditions imply that we never update $r'(y)$ such
that $r'(y) = r(y)$.

\begin{lemma}
	\label{lemma:dijkstra}
	After the Dijkstra-like algorithm terminates, $d'$ and $r'$ are correct.
\end{lemma}
A proof of this lemma can be found in Appendix~\ref{apx:proofs}.

\subsection{Variants and Algorithmic Improvements}
\label{sub:algorithm_eng}

\subsubsection{Semi-local Swaps}
One weakness of the algorithm in Section~\ref{sec:swapping}
is that it only performs local vertex exchanges. In particular, the
algorithm always swaps a vertex $u \in S$ and a vertex $v \in N(u) \setminus S$.
This condition can be generalized:
in particular, it is sufficient that $u \in S$ also satisfies
$u \in N(\add{S}v)$.
In this situation, the distances of all vertices can still only change by a single hop
and the proofs of the correctness of the algorithm remain valid.
Note that this naturally partitions candidates $u$ into two
sets: first, the set $N(v) \cap S$ of candidates that the original
algorithm considers, and the set $N(S) \cap S$.
Candidates in the latter set can be determined independently of $v$;
indeed, they can be swapped with any $v \in N(S) \setminus S$.
Hence, our swap selection strategy from
Section~\ref{sec:choose-swap} continues to work with little modification.

\subsubsection{Restricted Swaps}
\label{par:restricted_ls}
To further improve the performance of our Local-Swap algorithm
at the cost of its solution quality, we consider the following variant:
instead of selecting the pair of vertices $(u, v)$ that
maximize $|D_v| - |\Lambda_u|$, we just select the vertex $v$ that maximizes
$|D_v|$ and then choose $u \in N(v) \cap S$ such that $|\Lambda_u|$
is minimized. This restricts the choices for $u$; hence, we expect
this \emph{Restricted Local-Swap} algorithm to yield
solutions of worse quality. Due to the restriction, however,
it is also expected to converge faster.

\subsubsection{Local Grow-Shrink}
\label{par:local_gs}
During exploratory experiments, it turned out that the Grow-Shrink algorithm
sometimes overestimates the lower bound $|D_v|\cdot \dist{S, v}$
on the decrease $f(S) - f(\add{S}v)$ of the farness
after adding an element $v$. This happens because errors in the
approximation of $|D_v|$ are amplified by multiplying with a large
$\dist{S, v}$. Hence, we found that restricting the algorithm's choices
for $v$ to vertices near $S$ improves the
solution quality of the algorithm.

It may seem that this modification makes Grow-Shrink vulnerable to the
same weaknesses as Local-Swap. Namely, local choices
imply that large numbers of exchanges might be required to reach local optima
and it becomes hard to escape these local optima.
Fortunately, additional techniques discussed in the next section
can be used to avoid this problem.

\subsubsection{Extended Grow-Shrink}
\label{par:extended_gs}
Even in the case of Grow-Shrink,
the bound of Lemma~\ref{lemma:additions} becomes worse for vertices
at long distances from the current group.
As detailed in Section~\ref{sec:dagsize},
this happens as
our reachability set size approximation approach does not
take back edges into account.
This is a problem especially on graphs with a large diameter
where we have to expect that many back edges exist.
We mitigate this problem (as well as the problems
mentioned in Section~\ref{par:local_gs})
by allowing the group to grow by more
than one vertex before we shrink it again.
In particular, we allow the group to grow to size $k + h$ for some $h \geq 1$,
before we shrink it back to $k$.

In our experiments in Section~\ref{sec:experiments}, we consider two strategies to choose
$h$. First, we consider constant values for $h$.
However, this is not expected to be appropriate for all graphs:
specifically, we want to take the diameter of the graph into account.
Hence, a more sophisticated strategy selects
$h = \mathrm{diam}(G)/{k^p}$ for a fixed $p$.
This strategy is inspired by mesh-like graphs
(\eg real-world road networks and some other infrastructure networks):
if we divide a quadratic two-dimensional mesh $G$
into $k$ quadratic sub-meshes (where $k$ is a power of 2),
the diameter of the sub-meshes is $\mathrm{diam}(G)/\sqrt{k}$.
Hence, if we assume that each vertex of the group covers an equal
amount of vertices in the remaining graph, $h = \mathrm{diam}(G)/\sqrt{k}$
vertex additions should be sufficient to find at least one
good vertex that will improve a size-$k$ group.
As we expect that real-world networks deviate from ideal
two-dimensional meshes to some extend, we consider not only $p = \frac12$
but also other values of $p$.

\subsubsection{Engineering the reachability set size approximation algorithm}
\label{sub:reachability_engineering}
Cohen's reachability set size approximation algorithm~\cite{cohen1997size} has multiple parameters
that need to be chosen appropriately: in particular,
there is a choice of probability distribution (exponential vs. uniform),
the estimator function (averaging vs. selection-based),
the number of samples and the width of each random number.
For the estimator, we use an averaging estimator, as this estimator
can be implemented more efficiently than a selection-based estimator
(\ie it only requires averaging numbers instead of finding
the $k$-smallest number).
We performed exploratory experiments to determine a good configuration
of the remaining parameters.
It turns out that, while the exponential distribution
empirically offers better accuracy than the uniform distribution, the algorithm
can be implemented much more efficiently using the uniform distribution:
in particular, for the uniform distribution, it is sufficient to generate
and store the per-vertex random numbers as (unsigned) integers, while the exponential
distribution requires floating point calculations.
We deal with the decrease in accuracy by simply gathering more samples.
For the uniform distribution
and real-world graphs, 16 bits per integer turns out
to yield sufficient accuracy.
In this setting, we found that 16 samples
are enough to accurately find the vertex with highest reachability set size.
In particular, while the theoretical guarantee in~\cite{cohen1997size}
requires the number of samples to grow with $\log |V|$,
we found this number to have a negligible impact on the
solution quality of our group closeness heuristic
(see Appendix~\ref{apx:samples_exp}).

\subsubsection{Memory latency in reachability set size approximation}
It is well-known that the empirical performance of graph traversal
algorithms (like BFS and Dijkstra) is often limited by memory
latency~\cite{bader2005architectural, lumsdaine2007challenges}.
Unfortunately, the reachability set size approximation
needs to perform multiple traversals of the same graph.
To mitigate this issue, we perform multiple iterations
of the approximation algorithm at the same time.
This technique increases the running time performance
of the algorithm at the cost of its memory footprint.
More precisely, during each traversal of the graph,
we store 16 random integers per vertex and
we aggregate
all 16 minimal values per vertex at the same time.
This operation can be performed very efficiently by utilizing SIMD
vector operations.
In particular, we use 256-bit AVX operations of our Intel Xeon CPUs
to take the minimum of all 16 values at the same time.
As mentioned above, aggregating 16 random numbers per vertex
is enough for our use case; thus, using SIMD aggregation,
we only need to perform a single traversal of the graph.

\subsubsection{Accepting swaps and stopping condition}
As detailed in Sections~\ref{sec:swapping} and~\ref{sec:alternating},
our algorithms stop once they cannot find another vertex exchange
that improves the closeness score of the current group.
Exchanges that worsen the score are not accepted.
To prevent vertex exchanges that change the group closeness score
only negligibly, we also set a limit on the number of vertex exchanges.
In our experiments, we choose a conservative limit that does not
impact the solution quality measurably (see Appendix~\ref{apx:limit_impact}).

\section{Experiments}
\label{sec:experiments}
\todo{For final version: consider adding markers to bar plots.}
In this section, we evaluate the performance of our algorithms against
the state-of-the-art greedy algorithm of Bergamini
\etal~\cite{bergamini2018scaling}.\footnote{In our experiments we do
	not consider the naive greedy algorithm and
	the OSA heuristic of~\cite{chen2016efficient} because they are both
	dominated by~\cite{bergamini2018scaling}.}
As mentioned in Section~\ref{sec:intro}, it has been shown empirically
that the solution quality yielded by the greedy algorithm is often
nearly-optimal.
We evaluate two variants, \emph{LS} and \emph{LS-restrict} (see Section~\ref{par:restricted_ls}),
of our Local-Swap algorithm, and three
variants, \emph{GS}, \emph{GS-local} (see Section~\ref{par:local_gs})
and \emph{GS-extended} (see Section~\ref{par:extended_gs})
of our Grow-Shrink algorithm.
We evaluate these algorithms for group sizes
of $k \in \{5, 10, 20, 50, 100\}$ on the largest connected component
of the input graphs.
We measure the performance in terms of running time and closeness
of the group computed by the algorithms.
Because our algorithms construct an initial group $S$ by selecting $k$
vertices uniformly at random (see Section~\ref{sec:local_search}),
we average the results of five runs,
each one with a different random seed, using the geometric mean
over speedup and relative
closeness.\footnote{These five runs are done to average out
	particularly bad (or good) selections of initial groups;
	as one can see from Appendix~\ref{apx:samples_exp},
	the variance due to the randomized reachability set size
	algorithm is negligible.}
Unless stated otherwise, our experiments are based on the graphs
listed in
Tables~\ref{tab:inst_unw_table} and~\ref{tab:inst_w_table}.
They are all undirected and have been downloaded from
the public repositories 9th DIMACS Challenge~\cite{demetrescu2009shortest} and
KONECT~\cite{kunegis2013konect}.
The running time of the greedy baseline
varies between 10 minutes and 2 hours
on those instances.

Our algorithms are implemented in the
NetworKit~\cite{staudt2016networkit}
C++ framework and use
PCG32~\cite{pcg2014} to generate random numbers.
All experiments were managed by the SimexPal software to ensure
reproducibility~\cite{angrimanGLMNPT19}.
Experiments were executed with sequential code on a Linux machine with
an Intel Xeon Gold 6154 CPU and 1.5 TiB of memory.

\subsection{Results for Extended Grow-Shrink}
\label{sub:extended_exps}
\begin{figure}
\centering
\begin{subfigure}[t]{.041\columnwidth}
\centering
\includegraphics[angle=90]{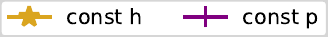}
\end{subfigure}\hfill
\begin{subfigure}[t]{.479\columnwidth}
\includegraphics{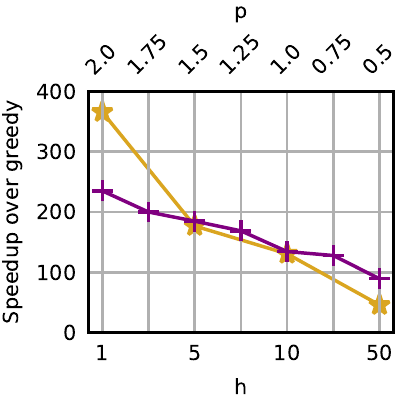}
\caption{Speedup over the greedy algorithm (geom. mean).}
\end{subfigure}\hfill
\begin{subfigure}[t]{.479\columnwidth}
\includegraphics{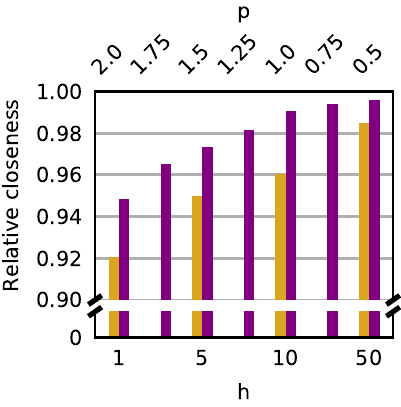}
\caption{Closeness score relative to the score of the group returned
	by greedy (geom. mean).}
\label{fig:quality_multi}
\end{subfigure}\hfill
\caption{Performance of the extended Grow-Shrink algorithm
	for different values of $h$ or $p$ on
	unweighted graphs, $k = 10$.}
\label{fig:plots_multi}
\end{figure}

In a first experiment, we evaluate
performance of our extended Grow-Shrink algorithm
and compare it to the greedy heuristic.
Because of its ability to escape local optima,
we expect this to be the best algorithm in terms of quality;
hence, it should be a good default choice among our algorithms.
For this experiment, we set $k = 10$.

As discussed in Section~\ref{par:extended_gs},
we distinguish two strategies to determine $h$: we either fix
a constant $h$, or we fix a constant $p$.
For each strategy, we evaluate multiple values for $h$ or $p$.
Results for both strategies are shown in Figure~\ref{fig:plots_multi}.
As expected, higher values of $h$ (or, similarly, lower values of $p$)
increase the algorithm's running time;
(while $h > 1$ allows the
algorithm to perform better choices, it does not converge $h$-times as fast).
Still, for all tested values of $h$ or $p$, the extended Grow-Shrink algorithm
is one to two orders of magnitude faster than the greedy baseline.
Furthermore, values of $p < 1$ yield results of very good quality:
for $p = 0.75$, for example, we achieve
a quality of $\onedigit{\gsPSevenFiveKQualNum}\%$.
At the same time, using this setting for $p$, our algorithm is
$\onedigit{\gsPSevenFiveKSpeed}\times$ faster than the greedy algorithm.
We remark that for all but the smallest values of $h$
(\ie those corresponding to the lowest quality), choosing constant
$p$ is a better strategy than choosing constant $h$:
for the same running time, constant $p$ always achieves solutions
of higher quality.

\subsection{Scalability to Large Graphs}
\label{sub:big_graphs}
We also analyze the running time of our extended Grow-Shrink algorithm on
large-scale networks. To this end, we switch to graphs larger than the ones
in Table~\ref{tab:inst_unw_table}. We fix $p = 0.75$, as
Section~\ref{sub:extended_exps} demonstrated that this setting results in
a favorable trade-off between solution quality and running time.
The greedy algorithm is not included in this experiment as it requires
multiple hours of running time, even for the smallest
real-world graphs that we consider in this part.
Hence, we also do not compare against its solution quality in
this experiment.

\subsubsection{Results on Synthetic Data}
\begin{figure}[t]
\centering
\begin{subfigure}[t]{.5\columnwidth}
\centering
\includegraphics{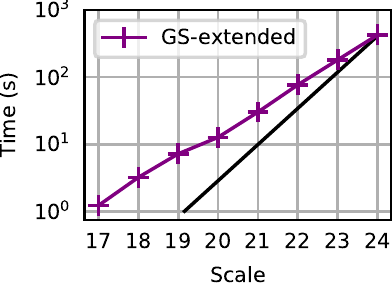}
\caption{R-MAT networks; $2^{17}$
	to $2^{24}$ vertices (up to $268$ million edges).}
\label{fig:run_time_rmat}
\end{subfigure}%
\hfill
\begin{subfigure}[t]{.5\columnwidth}
\centering
\includegraphics{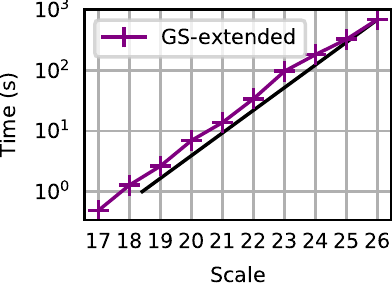}
\caption{Random hyperbolic networks; $2^{17}$
	to $2^{26}$ vertices (up to $671$ million edges).}
\label{fig:run_time_hyp}
\end{subfigure}
\caption{Running time (s) of the extended Grow-Shrink algorithm
	on synthetic graphs (black line = linear regression); $p = 0.75$, $k = 10$.}
\label{fig:synthetic_insts}
\end{figure}

Figure~\ref{fig:synthetic_insts} shows the average running time
of our algorithm
on randomly generated R-MAT~\cite{chakrabarti2004r} graphs as well as
graphs from a generator~\cite{DBLP:conf/hpec/LoozOLM16} for random hyperbolic graphs.
Like R-MAT, the random hyperbolic model yields graphs with a skewed degree
distribution, similar to the one found in real-world complex networks.
In the (log-log) plot, the straight lines represent a linear regression of the running times.
In both cases, the running time curves are at most as
steep as the regression line, \ie the running time behaves linearly
in the number of vertices
for the considered network models and sizes.

\subsubsection{Results on Large Real-World Data Sets}
\begin{table}[tb]
\caption{Running time of the extended Grow-Shrink algorithm on
	large real-world networks; $p = 0.75$, $k = 10$.}
\label{tab:large_insts}
\centering
\begin{tabular}{lrrr}
Network & $|V|$ & $|E|$ & Time (s) \\
\toprule
soc-LiveJournal1 & \numprint{4843953} & \numprint{42845684} & \numprint{95.3}\\
livejournal-links & \numprint{5189808} & \numprint{48687945} & \numprint{135.6}\\
orkut-links & \numprint{3072441} & \numprint{117184899} & \numprint{199.9}\\
dbpedia-link & \numprint{18265512} & \numprint{126888089} & \numprint{368.0}\\
dimacs10-uk-2002 & \numprint{18459128} & \numprint{261556721} & \numprint{333.1}\\
wikipedia\_link\_en & \numprint{13591759} & \numprint{334640259} & \numprint{680.1}\\
\bottomrule
\end{tabular}

\end{table}

Table~\ref{tab:large_insts} reports the algorithm's performance
on large real-world graphs.
In contrast to the greedy algorithm (which would require hours),
our extended Grow-Shrink algorithm can handle real-world graphs with
hundreds of millions of edges in a few minutes.
For the orkut-links network, Bergamini \etal~\cite{bergamini2018scaling} report
running times for greedy of 16 \emph{hours} on their machine;
it is the largest instance in their experiments.

\subsection{Accelerating Performance on Unweighted Graphs}
\label{par:unw_graphs}
\begin{figure}[t]
\centering
\begin{subfigure}[t]{\columnwidth}
\centering
\includegraphics{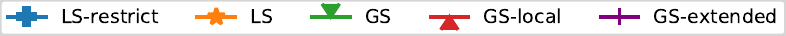}
\end{subfigure}
\centering
\begin{subfigure}[t]{.5\columnwidth}
\centering
\includegraphics{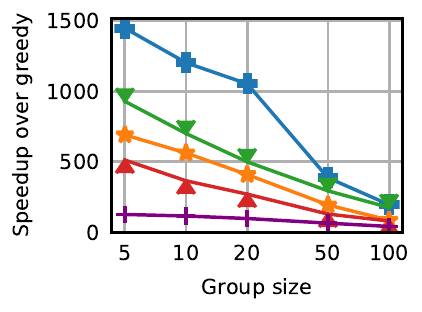}
\caption{Speedup over the greedy algorithm (geom. mean).}
\label{fig:speedup}
\end{subfigure}\hfill
\begin{subfigure}[t]{.5\columnwidth}
\centering
\includegraphics{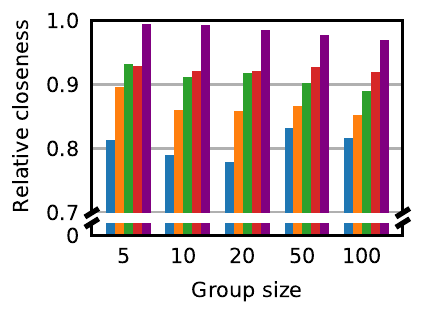}
\caption{Closeness score relative to the score of the group returned by greedy (geom. mean).}
\label{fig:quality}
\end{subfigure}
\caption{Performance of our local search algorithms for different values of $k$;
	unweighted graphs.}
\label{fig:plots}
\end{figure}

While the extended Grow-Shrink algorithm yields results of very high quality,
if quality is not a primary concern,
even faster algorithms might be desirable for very large graphs.
To this end, we also evaluate the performance of the non-extended
Grow-Shrink and the Local Swap algorithms.
For extended Grow-Shrink, we fix $p = 0.75$ again.
The speedup and the quality of our algorithms over the greedy algorithm,
for different values of the group size $k$,
is shown in Figures~\ref{fig:speedup} and~\ref{fig:quality},
respectively.
Note that the greedy algorithm scales
well to large $k$,
so that the speedup of our algorithms decreases with
$k$ (as mentioned in Section~\ref{sec:intro}, the main bottleneck
	of greedy is adding the first vertex to the group).
However, even for large groups
of $k = 100$, all of our algorithms are still at least
$\onedigit{\minSpeedupKHundred}\times$ faster.

After extended Grow-Shrink, our non-extended local version of Grow-Shrink is the next
best algorithm in terms of quality.
As explained in Section~\ref{par:local_gs},
this variant yields better solutions than non-local Grow-Shrink
and gives a speedup of $\onedigit{\gsLocalSpeedupOverNonLocal}\times$
over extended Grow-Shrink with $p = 0.75$ and $k = 10$
(= a speedup of $\onedigit{\gsLocalSpeedKTen}\times$ over greedy);
the solution quality in this case is $\onedigit{\gsLocalQualKTen}\%$
of the greedy quality.

The non-restricted Local-Swap algorithm is dominated by Grow-Shrink,
both regarding running time and solution quality.
Furthermore, compared to the other algorithms,
the restricted Local-Swap algorithm only gives a rough estimate of the group
with highest group closeness; however, it is also significantly
faster than all other algorithms and may be employed
during an exploratory analysis of graph data sets.

\subsection{Results on Weighted Road Networks}
\begin{figure}[t]
\centering
\begin{subfigure}[t]{\columnwidth}
\centering
\includegraphics{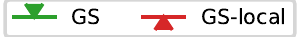}
\end{subfigure}
\centering
\begin{subfigure}[t]{.5\columnwidth}
\centering
\includegraphics{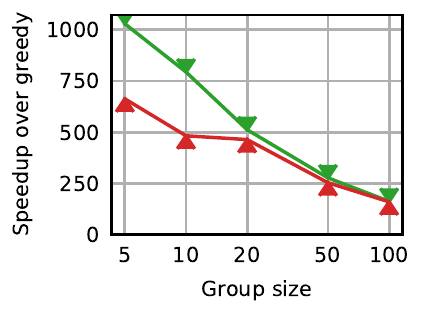}
\caption{Speedup over the greedy algorithm (geom. mean).}
\label{fig:speedup_weighted}
\end{subfigure}%
\hfill
\begin{subfigure}[t]{.5\columnwidth}
\centering
\includegraphics{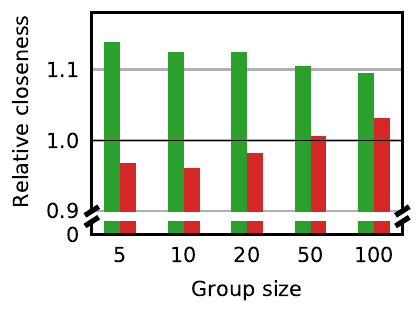}
\caption{Closeness score relative to the score of the group returned by greedy (geom. mean).}
\label{fig:quality_weighted}
\end{subfigure}
\caption{Performance of our local search algorithms for different values of $k$;
	weighted graphs.}
\label{fig:plots_weighted}
\end{figure}

Recall that the Local-Swaps algorithm does not work for weighted graphs;
we thus report only Grow-Shrink data in the weighted case.
The performance of Grow-Shrink and local Grow-Shrink on weighted graphs is shown in
Figure~\ref{fig:plots_weighted}.
In contrast to unweighted graphs, the quality of the non-local Grow-Shrink algorithm is superior to
the greedy baseline for all considered group sizes.
Furthermore, in contrast to the unweighted case, the ability to
perform non-local vertex exchanges greatly benefits the
non-local Grow-Shrink algorithm
compared to local
Grow-Shrink.\footnote{For this reason, we do not include the extended Grow-Shrink algorithm
	in this experiment. In fact, we expect that it improves only slightly on GS-local (red line/bars
	in Fig.~\ref{fig:plots_weighted})
	but cannot compete with (non-local) GS:
the ability to perform non-local vertex exchanges, as done by GS (green line/bars in Fig.~\ref{fig:plots_weighted})
	appears to be crucial to obtain high-quality results on weighted graphs.}
Thus, Grow-Shrink clearly dominates both the greedy algorithm and
local Grow-Shrink on the weighted graphs
in our benchmark set -- both in terms of speed \emph{and} solution quality.

\subsection{Summary of Experimental Results}
On unweighted graphs,
a good trade-off between solution quality and running time
is achieved by the extended Grow-Shrink algorithm with constant $p = 0.75$.
This strategy yields solutions with at least $\onedigit{\gsPSevenFiveKQualNum}\%$
of the closeness score of a greedy solution (greedy, in turn, was at most $3\%$
away from the optimum on small networks in previous work~\cite{bergamini2018scaling}).
Extended Grow-Shrink is $\onedigit{\gsPSevenFiveKSpeed}\times$ faster than
greedy ($k = 10$). Thus, it is able to handle graphs with
hundreds of millions of edges in a few minutes -- the state of the art needs multiple hours.
If a fast but inaccurate algorithm is needed for exploratory analysis of graph data sets,
we recommend to
run the non-extended Grow-Shrink algorithm, or, if only a very coarse estimate
of the group with maximal closeness is needed, restricted Local-Swap.

On weighted graphs, we recommend to always use our Grow-Shrink algorithm,
as it outperforms the greedy state of the art both in terms of quality
(yielding solution that are on average $\onedigit{\wGSQualInc}\%$ better than greedy solutions)
and in terms of running time performance (with a speedup of two orders of magnitude),
at the same time.

\section{Conclusions}

In this paper, we introduced two families of new local-search algorithms for group closeness maximization in
large networks.
As maximizing group closeness exactly is infeasible for graphs with more than a few thousand edges,
our algorithms are heuristics (just like the state-of-the-art greedy algorithm).
However, for small real-world networks, the results
are empirically known to be close to optimal
solutions~\cite{bergamini2018scaling}.

Compared to previous state-of-the-art heuristics,
our algorithms (in particular: extended Grow-Shrink)
allow to find groups with high closeness
centrality in real-world networks with hundreds of millions of edges
in seconds to minutes instead of multiple hours,
while sacrificing less than $1\%$ in quality.
In weighted graphs, Grow-Shrink (GS) even dominates the best
known heuristic: the GS solution quality is more than $10\%$ higher
and GS is two orders of magnitude faster.
Adapting the algorithm to even larger graphs in distributed memory
is left to future work.

\bibliographystyle{IEEEtran}
\bibliography{IEEEabrv,biblio}

\appendices
\section{Technical Proofs}
\label{apx:proofs}

\begin{proof}[Proof of Lemma~\ref{lemma:additions}]
	Let $x \in D_v$.
	Because of the sub path optimality property of shortest paths,
	it is clear that $\mathrm{dist}(S, x)
	= \mathrm{dist}(S, v) + \mathrm{dist}(v, x)$
	(as $v$ is a predecessor of $x$ on a shortest path from $S$).
	On the other hand, adding $v$ to $S$ decreases the length of this path
	(as the distance between $S$ and $v$ becomes zero);
	in other words: $\dist{\add{S}v, x} = \mathrm{dist}(v, x)$.
	These observations allow us to express the right-hand side of Eq.~\ref{eq:delta-minus}
	as $\mathrm{dist}(S, x) - \mathrm{dist}(\add{S}v, x) = \mathrm{dist}(S, v)$.
	Summing this equation for all vertices in $D_v$ yields
	the term $|D_v|\cdot \dist{S, v}$ of the lemma.
	For vertices $x \notin D_v$, it holds that
	$\dist{S, x} - \dist{\add{S}v, x} \geq 0$,
	hence the inequality.

	For the statement about the unweighted case,
	we need to show that the contribution of all other vertices is zero,
	\ie $\mathrm{dist}(S, x) = \mathrm{dist}(\add{S}v, x)$
	for all vertices $x \notin D_v$.
	Note that
	$\mathrm{dist}(S, x) < \mathrm{dist}(S, v) + \mathrm{dist}(v, x)$
	(otherwise $x$ would be in $D_v$) and $\mathrm{dist}(S, v) = 1$. Thus,
	$\mathrm{dist}(S, x) \leq \mathrm{dist}(v, x)$ which
	completes the proof.
\end{proof}

\begin{proof}[Proof of Lemma~\ref{lemma:dijkstra}]
	Let $Z$ be the set of vertices that need to be updated by the
	algorithm, \ie $Z$ equals the set $R_u \cup R'_u$ before
	the Dijkstra-like algorithm runs.
	We have not shown yet whether all vertices in $Z$ are indeed updated.
	For the remainder of this proof, all symbols
	(such as $r'$, $d'$ and $R'$)
	refer to the state of our data structures
	after the algorithm terminates.
	To prove the lemma, it is sufficient to prove that
	$Z = \bigcup_{w \in S} (Z \cap R'_w)$ (\ie that no $r'$ remains
	undefined, or, in other words, $r'$ is updated wherever necessary)
	and that
	$\dist{\remove{S}{r(x)}, x} = \dist{r'(x), x}$ for all $x \in Z$
	(\ie that the definition of $r'$ is respected).

	Let us first prove that $\bigcup_{w \in S} (Z \cap R'_w) = Z$.
	Let $z \in Z$. There exists a path from every $w \in S$ to $z$
	and each such path contained at least one boundary pair $(x, y)$
	before the algorithm started. Indeed, there is a boundary pair for
	the first vertex $y$ on that path that is also in $Z$.
	Thus, the algorithm sets $r'(y) = w$ for some $w \in S$
	(\ie $y \in Z \cap R'_w$) and propagates the update of $r'$
	along the path from $w$ to $z$.
	We have to prove that our pruning condition does not
	prevent any necessary update along this path.
	Hence, let $(x, y)$ be a pair of vertices so that
	the algorithm is pruned before visiting $y$ from $x$.
	Only the $y \in Z$ case is interesting, as $r'$ must already be correct otherwise.
	Pruning only happens if $r'(x) = r(y)$ and therefore $r(x) \neq r(y)$.
	But in this case, $(x, y)$ was a $d$-boundary pair
	and the preceding argument shows that $y \in \bigcup_{w \in S} Z \cap R'_w$.

	Now consider the second part of the proof.
	Let $w \in S$ be any group vertex and let $y \in Z \cap R_w$
	be any vertex that is updated by the algorithm with $r(y) = w$.
	The algorithm guarantees that
	$\dist{\remove{S}w, y} \leq \dist{r'(y), y} \leq d'(y)$
	as $r'(y) \in \remove{S}w$ and
	(by construction of the algorithm) $d'$ is the length
	of a (not yet proven to be shortest) path from $r'(y)$ to $y$.
	It is sufficient to show that this path is a shortest one,
	\ie $\dist{\remove{S}w, y} = d'(y)$.
	We prove this statement for all $y \in R_w$ by
	an inductive argument using $\dist{\remove{S}w, y}$.
	We distinguish two cases depending on whether
	there exists a neighbor of $y$ in $R_w$ that
	is on a shortest path from $\remove{S}w$ to $y$.
	First, we handle the case that no such neighbor exists.
	In this case, $r(x) \neq w$ holds for all $x \in N(y)$
	on shortest paths from $\remove{S}w$ to $y$.
	As $r(x)$ did not change during the algorithm,
	all such $x$ correspond to
	$d$-boundary pairs for $y$
	and $\dist{\remove{S}w, y}$ is the minimal boundary distance
	over all these pairs $(x, y)$.
	Hence, $d'(y)$ was updated correctly before the Dijkstra-like algorithm ran.
	On the other hand, let $x \in N(y) \cap R_w$ be a neighbor
	of $y$ that is on a shortest path from $\remove{S}w$ to $y$.
	$x \in R_w$ implies $r'(x) \neq w$; thus, the algorithm cannot
	be pruned when visiting $y$ from $x$.
	In this case, however, the algorithm sets $d'(y) = d'(x) + \dist{x, y}$.
	As $\dist{\remove{S}w, x} < \dist{\remove{S}w, y}$,
	the induction yields that $d'(x)$ is already correct, \ie $d'(x) = \dist{\remove{S}w, x}$.
	Since $x$ is on a shortest path from $\remove{S}w$ to $y$,
	$d'(y)$ is also updated correctly.
\end{proof}

\section{Details of the Experimental Setup}
\label{apx:experimental_details}
\begin{table}[h]
\centering
\caption{Unweighted real-world networks used in the experiments.}
\label{tab:inst_unw_table}
\begin{tabular}{lrrr}
Network & $|V|$ & $|E|$ & Category \\
\toprule
dimacs9-NY & \numprint{264346} & \numprint{365050} & Road\\
dimacs9-BAY & \numprint{321270} & \numprint{397415} & Road\\
web-Stanford & \numprint{255265} & \numprint{1941926} & Hyperlink\\
hyves & \numprint{1402673} & \numprint{2777419} & Social\\
youtube-links & \numprint{1134885} & \numprint{2987468} & Social\\
com-youtube & \numprint{1134890} & \numprint{2987624} & Social\\
web-Google & \numprint{855802} & \numprint{4291352} & Hyperlink\\
trec-wt10g & \numprint{1458316} & \numprint{6225033} & Hyperlink\\
dimacs10-eu-2005 & \numprint{862664} & \numprint{16138468} & Road\\
soc-pokec-relationships & \numprint{1632803} & \numprint{22301964} & Social\\
wikipedia\_link\_ca & \numprint{926588} & \numprint{27133794} & Hyperlink\\
\bottomrule
\end{tabular}

\end{table}
\begin{table}[h]
\centering
\caption{Weighted networks used in the experiments.
	All networks are road networks of different states of the US.}
\label{tab:inst_w_table}
\begin{tabular}{lrr}
State & $|V|$ & $|E|$ \\
\toprule
DC & \numprint{9522} & \numprint{14807}\\
HI & \numprint{21774} & \numprint{26007}\\
AK & \numprint{48560} & \numprint{55014}\\
DE & \numprint{48812} & \numprint{59502}\\
RI & \numprint{51642} & \numprint{66650}\\
CT & \numprint{152036} & \numprint{184393}\\
ME & \numprint{187315} & \numprint{206176}\\
\bottomrule
\end{tabular}
\hfill
\begin{tabular}{lrr}
State & $|V|$ & $|E|$ \\
\toprule
ND & \numprint{203583} & \numprint{249809}\\
SD & \numprint{206998} & \numprint{249828}\\
WY & \numprint{243545} & \numprint{293825}\\
ID & \numprint{265552} & \numprint{310684}\\
MD & \numprint{264378} & \numprint{312977}\\
WV & \numprint{292557} & \numprint{320708}\\
NE & \numprint{304335} & \numprint{380004}\\
\bottomrule
\end{tabular}

\end{table}

Tables~\ref{tab:inst_unw_table} and ~\ref{tab:inst_w_table} show details about our real-world
instances.
To generate the synthetic graphs in Figure~\ref{fig:synthetic_insts}, we use the same parameter
setting as in the Graph 500's benchmark~\cite{murphy2010introducing} (\ie edge factor 16,
$a = 0.57$, $b = 0.19$, $c = 0.19$,
and $d = 0.05$) for the R-MAT generator. For the random hyperbolic generator,
we set the average degree to $20$, and the exponent of the power-law distribution to 3.

\section{Additional Experiments}

\subsection{Impact of the number of vertex exchanges}
\label{apx:limit_impact}

Figures~\ref{fig:farness_impr} and~\ref{fig:farness_impr_weighted}
depict the relative closeness (compared to the closeness of the group returned by the
greedy algorithm), depending on the progress of the algorithm in
terms of vertex exchanges.
For extended Grow-Shrink, we fix $p = 0.75$.
All of the local search algorithm quickly converge to a value near their
final result; additional vertex exchanges improve the group closeness
score by small amounts.
In order to avoid an excessive amount of iterations,
it seems reasonable to set a limit on the number of
vertex exchanges.
In our experiments we set a conservative limit of 100 exchanges.

\subsection{Impact of reachability set size approximation}
\label{apx:samples_exp}
\begin{figure}
\centering
\begin{subfigure}[t]{\columnwidth}
\centering
\includegraphics{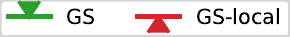}
\end{subfigure}
\begin{subfigure}[t]{.5\columnwidth}
\includegraphics{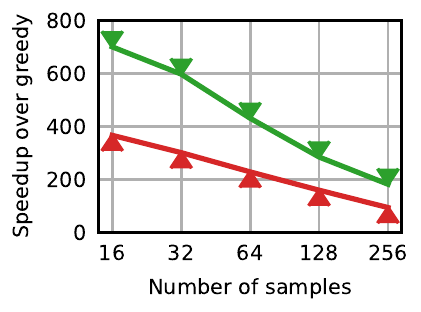}
\caption{Speedup over the greedy algorithm (geom. mean).}
\label{fig:speedup_dagiter}
\end{subfigure}\hfill
\begin{subfigure}[t]{.5\columnwidth}
\includegraphics{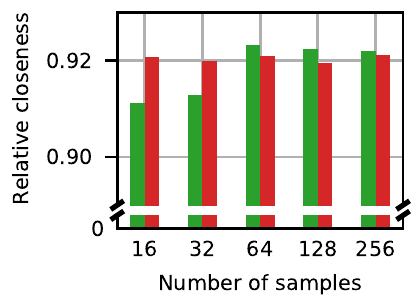}
\caption{Closeness score relative to the score of the group returned by greedy (geom. mean).}
\label{fig:quality_dagiter}
\end{subfigure}\hfill
\caption{Performance of the Grow-Shrink
algorithm for different numbers of samples to estimate
reachability set size; $k = 10$.}
\label{fig:plots_dagiter}
\end{figure}

As mentioned in Section~\ref{par:local_gs}, the errors in the approximation of
$|D_v|$ are amplified by the multiplication with $\dist{S,v}$, and this
results in GS-local computing higher quality solutions than GS.
We study how increasing the accuracy of the
reachability set size approximation by incrementing the
number of samples impacts the performances of both GS
and GS-local.
Figure~\ref{fig:speedup_dagiter} shows that GS needs at least 64
samples to converge to a better local optimum than GS-local.
However, in both cases increasing the number of samples
degrades the speedup without yielding
a significant quality improvement
(Figure~\ref{fig:quality_dagiter}).

\begin{figure}
\centering
\begin{subfigure}[t]{\columnwidth}
\centering
\includegraphics{plots/legend}
\end{subfigure}
\centering
\begin{subfigure}[t]{.5\columnwidth}
\centering
\includegraphics{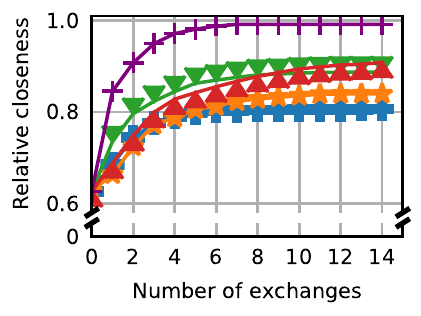}
\caption{Unweighted graphs.}
\label{fig:farness_impr}
\end{subfigure}\hfill
\begin{subfigure}[t]{.5\columnwidth}
\centering
\includegraphics{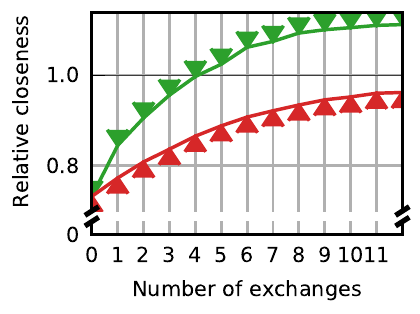}
\caption{Weighted graphs.}
\label{fig:farness_impr_weighted}
\end{subfigure}
\caption{Behavior of the relative closeness score
	(compared to the group returned by greedy, geom. mean)
	over the execution of the algorithms (in terms of vertex exchanges); $k = 10$.}
\end{figure}

\end{document}